\documentclass[aps,pra,superscriptaddress,twocolumn]{revtex4}

\usepackage{amsmath,amssymb,amsthm,graphicx}
\usepackage{color}

\begin{document}


\title{Hybrid classical-quantum linear solver using Noisy Intermediate-Scale Quantum machines}

\author{Chih-Chieh Chen}
\email[]{itriA70068@itri.org.tw}
\affiliation{Electronic and Optoelectronic System Research Laboratories, Industrial Technology Research Institute, Hsinchu 31057, Taiwan}

\author{Shiue-Yuan Shiau}
\affiliation{Physics Division, National Center for Theoretical Sciences, Hsinchu 30013, Taiwan}

\author{Ming-Feng Wu}
\affiliation{Electronic and Optoelectronic System Research Laboratories, Industrial Technology Research Institute, Hsinchu 31057, Taiwan}

\author{Yuh-Renn Wu}
\affiliation{Graduate Institute of Photonics and Optoelectronics and Department of Electrical Engineering, National Taiwan University, Taipei 10617, Taiwan}
\affiliation{Electronic and Optoelectronic System Research Laboratories, Industrial Technology Research Institute, Hsinchu 31057, Taiwan}


\date{\today}

\begin{abstract}
 We propose a realistic hybrid classical-quantum linear solver   to solve systems of linear equations of a specific type, and demonstrate its feasibility  with Qiskit on IBM Q systems. This  algorithm  makes use of quantum random walk that runs in $\mathcal{O}(N\log(N))$ time on a quantum circuit made of $\mathcal{O}(\log(N))$ qubits. The input and output are  classical data, and so can be easily accessed. It  is   robust against noise, and ready for implementation in applications such as machine learning.  
\end{abstract}


\maketitle


\section{Introduction}

Algorithms that run on quantum computers hold promise to perform important computational tasks more efficiently than what can ever be achieved on classical computers, most notably Grover's search algorithm and Shor's integer factorization \cite{Nielsen2011}.  One   computational task indispensable for many problems in science, engineering, mathematics, finance, and machine learning, is solving systems of linear equations $\textbf{A}\vec{x}=\vec{b}$. Classical direct and iterative algorithms take $\mathcal{O}(N^3)$ and $\mathcal{O}(N^2)$ time \cite{Golub1996,Saad2003}. Interestingly, the Harrow-Hassidim-Lloyd (HHL) quantum algorithm \cite{Harrow2009,Clader2013,Montanaro2016,Childs2017,Costa2017,Berry2017,Dervovic2018,Wossnig2018,Biamonte2017,Ciliberto2018}, which is based on the quantum circuit model \cite{Deutsch1985}, takes only  $\mathcal{O}(\log(N))$  to solve a sparse $N\times N$ system of linear equations, while for dense systems it requires $\mathcal{O}(\sqrt{N}\log(N))$ \cite{Wossnig2018}. Linear solvers and experimental realizations that use  quantum annealing and adiabatic quantum computing machines \cite{Kadowaki1998,Farhi2000,Aharonov2008} are also reported \cite{OMalley2016,Borle2018,Chang2018}. Most recently, methods \cite{Wen2018,Subasi2018} inspired by adiabatic quantum computing are proposed to be implemented on circuit-based quantum computers.  Whether substantial quantum speedup exists in these algorithms remains unknown.

In practice, the applicability of quantum algorithms to  classical systems are limited by the short coherence time of noisy quantum hardware in the so-called Noisy Intermediate-Scale Quantum (NISQ) era \cite{Preskill2018} and the difficulty in executing the input and output of classical data. Other roadblocks toward practical implementation include limited number of qubits, limited connectivity between qubits, and large error correction overhead. At present, experiments demonstrating the HHL linear solver on circuit quantum computers are limited to $2\times 2$ matrices \cite{Cao2012,Cai2013,Barz2014,Pan2014,Zheng2017,Lee2018},  while  linear solvers   inspired by adiabatic quantum computing are limited to $8\times 8$ matrices \cite{Wen2018,Subasi2018}. For quantum annealers, the state-of-the-art linear solvers can solve up to $12\times 12$ matrices \cite{Chang2018}.

In addition to the problems of limited available entangled qubits and short coherence time, the HHL-type algorithms are designed to work  only when  input and output are  quantum states \cite{Aaronson2015}. This condition imposes severe restriction to practical applications in the NISQ era \cite{Childs2009,Aaronson2015,Preskill2018}. It has been shown that the HHL algorithm can not extract information about the norm of the solution vector $\vec{x}$ \cite{Harrow2009}. A state preparation algorithm for inputting a classical vector $\vec{b}$ would take $\mathcal{O}(N)$ time \cite{Mottonen2004,Plesch2011,Aaronson2015,Coles2018}, with large overhead for current hardware. In addition, quantum state tomography is required  to read out the classical solution vector $\vec{x}$,  which is a demanding task \cite{James2001,Suess2017}, except  for special cases like  one-dimensional entangled qubits \cite{Cramer2010}. Inputting the matrix $\textbf{A}$ is also a challenge that may kill the quantum speedup \cite{Nielsen2011,Cao2012,Cai2013,Barz2014,Pan2014,Zheng2017,Lee2018}.

 In this work, we propose a  hybrid classical-quantum  linear solver that uses circuit-based quantum computer to perform quantum random walks. In contrast to the  HHL-type linear solvers, the solution vector $\vec{x}$ and the constant vector $\vec{b}$   in this hybrid algorithm stay as classical data in the classical registers. Only the matrix $\textbf{A}$ is encoded in quantum registers.  The idea is similar to that of variational quantum eigensolvers \cite{Peruzzo2014,Wecker2015,McClean2016,Kandala2017}, where  quantum speedup  is exploited only for sampling  exponentially large state Hilbert spaces, while the rest of  computational task is done by classical computer. This makes it easy to perform  data input and output: the $\vec{b}$ vector can be arbitrary, and the components and the norm of  the $\vec{x}$ vector can be easily accessed. \

We consider matrices  that are useful for Markov decision problems such as in reinforcement learning \cite{Sutton1998}. We show that  these matrices can be efficiently encoded by introducing the Hamming cube structure: a square matrix of size $N$ requires $\mathcal{O}(\log(N))$ quantum bits only. The quantum random walk algorithm we here propose takes  $\mathcal{O}(\log(N))$ time to obtain one component of the $\vec{x}$ vector. We also show that in the quantum random walk algorithm  the matrices produced as a result of qubit-qubit correlation  are inherently complex, which can be an advantage for performing difficult tasks. For  the same amount of time, the matrices the classical random walk algorithm can solve are limited to  factorisable ones only.

We have tested the quantum random walk algorithm using software development kit Qiskit on IBM Q systems \cite{IBMQ2017qiskit,IBMQ2017}. Numerical results show that this  linear solver  works  on ideal quantum computer, and most importantly, also on noisy quantum computer having a  short coherence time, provided the  quantum circuit that encodes the $\textbf{A}$ matrix  is not too long.    The limitation due to machine errors is discussed.

\section{Methods}

\begin{figure*}[t]	
	\includegraphics[width=\textwidth]{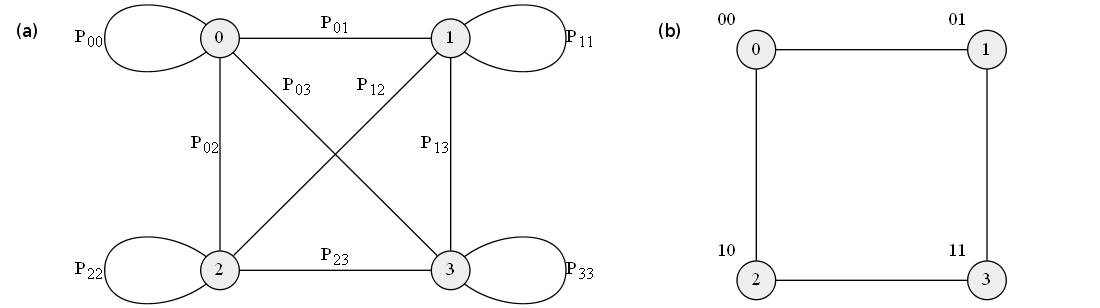}\\
	\caption{(a) Quantum (or classical) random walk on an undirected $N=4$ graph. The  transition probability of going from node $I$ to node $J$ or vice versa is equal to $P_{I,J}$, these elements forming a $4\times4$ matrix. (b) The four nodes on this Hamming cube are labeled by integers $(0,1,2,3)$; they are encoded as four different states $|00\rangle$, $|01\rangle$, $|10\rangle$, $|11\rangle$, respectively.  \label{rw}}
\end{figure*}

We consider a system of linear equations of real numbers $\textbf{A}\vec{x}=\vec{b}$, where $\textbf{A}$ is a $N\times N$ matrix to be solved, $N\times1$ vectors $\vec{x}$ and $\vec{b}$ are, respectively, the solution vector and a vector of constants. Without loss of generality, we rewrite $\textbf{A}$ as 
\begin{equation}
\textbf{A}=\mathbf{1}-\gamma \textbf{P}\,,\label{matrxA}
\end{equation}
 where $\mathbf{1}$ is the identity matrix, and $ 0 < \gamma < 1$ is a real number. 
 We take $\textbf{P}$  as a (stochastic) Markov-chain transition matrix, such that $P_{I,J} \ge 0$ and $\sum_J P_{I,J}=1$, where $P_{I,J}$ refers to the $\textbf{P}$ matrix element in the $J$-th column of the $I$-th row. This type of linear systems appears in value estimation for reinforcement learning \cite{Barto1993,Sutton1998,Paparo2014}, and radiosity equation in computer graphics \cite{Goral1984}. In reinforcement learning algorithms, given a fixed policy of the learning agency, the vector $\vec{x}$ is the value function that determines the long-term cumulative reward, and efficient estimation of this function is  key to successful  learning \cite{Sutton1998}.  Note that the matrix $\textbf{A}$ given in Eq.~(\ref{matrxA}) used as model Hamiltonian matrix  belongs to the so-called  stoquastic Hamiltonians \cite{Bravyi2006,Bravyi2015}.   \

To solve $\textbf{A}\vec{x}=\vec{b}$, we expand the solution vector as Neumann series, that is, $\vec{x}=\textbf{A}^{-1}\vec{b}=(\mathbf{1}-\gamma \textbf{P})^{-1}\vec{b}=\sum_{s=0}^\infty \gamma^s \textbf{P}^s \vec{b}$. Let us define the $I_0$ component of $\vec{x}$ truncated up to $\gamma^c$ terms as 
\begin{equation}
x^{(c)}_{I_0}=\sum_{s=0}^c \gamma^s \sum_{I_1,...,I_s=1}^N P_{I_0,I_1}...P_{I_{s-1},I_s} b_{I_s}\,. \label{expand}
\end{equation}
 This expression for $x^{(c)}_{I_0}$  can be evaluated by random walks  on a graph of $N$ nodes, with the probability of going from  node $I$ and node $J$ of the graph given by the matrix element $P_{I,J}$, which we set as symmetric (undirected), namely $P_{I,J}=P_{J,I}$. An example of a four-node graph is shown in Fig.~\ref{rw}(a).  By performing a series of random walks starting from node $I_0$, walking $c$ steps according to the transition probability matrix $\textbf{P}$, and ending at some node $I_c$, Eq.~(\ref{expand}) can be readily calculated to get the $x^{(c)}_{I_0}$ value, which is close to the solution $x_{I_0}$ for some large $c$ steps.  Truncating the series introduces an error $\epsilon\sim\mathcal{O}(\gamma^c)$. So, for a given $\gamma$, the number of steps necessary to meet a given tolerance $\epsilon$ is equal to  $c\sim \log (1/\epsilon)/\log (1/\gamma)$.

The above expansion procedure can be extended to more general matrices $\textbf{A}$ by setting $\textbf{A}=\mathbf{1}-\textbf{B}$ where $B_{I,J}=P_{I,J}v_{I,J}$ for  real matrix elements $v_{I,J}$ provided that the eigenvalues of $\textbf{B}$ are bounded by $-1<\mbox{eig}(\textbf{B})_I<1$.

For classical Monte Carlo methods to compute Eq.~(\ref{expand}), it takes  $\mathcal{O}(N)$ time to calculate the cumulative distribution function that is used to determine the next walking step. So, these linear systems   can be solved by classical Monte Carlo methods within $\mathcal{O}(N^2)$ time \cite{Metropolis1949,Forsythe1950,Wasow1952,Lu2003,Branford2008}. Similar Monte Carlo  methods have been extended to more general matrices for applications in Green's function Monte Carlo method for many-body physics \cite{Ceperley1980,Negele1988,Landau2005}.

\subsection{Encoding  state spaces on Hamming cubes}
As for material resources, in general it takes at least $\mathcal{O}(N)$ classical bits to store a row of a  stochastic transition matrix $\textbf{P}$ (or $\textbf{A}$). However, for the classical and quantum random walks we here consider, it is possible to reduce significantly the number of classical or quantum bits necessary to encode the corresponding transition probability matrix $\textbf{P}$ to $\mathcal{O}(\log(N))$ by introducing the Hamming cube (HC) structure \cite{Hamming1950}.
To do it, we first associate each graph node with a bit string. As shown in Fig.~\ref{rw}(b), the four nodes of the $N=4$ graph are fully represented by  two bits.  Node states $|0\rangle$, $|1\rangle$, $|2\rangle$, and $|3\rangle$ represent binary string states $|00\rangle$, $|01\rangle$, $|10\rangle$, and $|11\rangle$, respectively. For a $N$-node graph,  only $\log_2(N)= n$ (to base 2)  bits are needed to encode the integers $J\in \{0,1,..., N{-}1\}$, each representing the $n$-bit binary string state, namely  $|J\rangle=|j_{n-1},...j_1,j_0\rangle$, where $j_\ell$ is 0 or 1.

 \begin{figure*}[t]
 	\includegraphics[width=\textwidth]{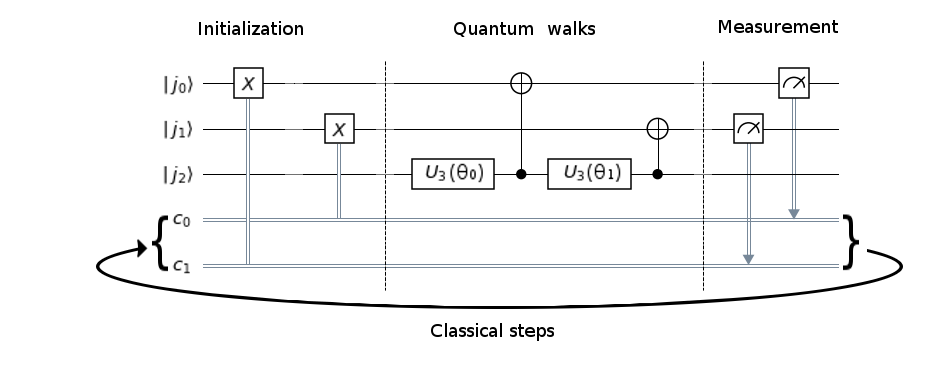}
 	\caption{Discrete-time coined quantum walk circuit for the $4\times 4$ transition matrix given in Eq.~(\ref{pmatrix4}). Qubits $j_0$ and $j_1$ are state register qubits to represent the four-node graph in Fig.~\ref{rw}, first set as $0$ before initialization, while the qubit $j_2$ is the coin register qubit.  The measured registers $c_0$ and $c_1$ are fed back to initialize the next iteration. The classical-step  is repeated $c$ times to obtain the Neumann expansion up to order $c$. } \label{circuit}
 \end{figure*}

\subsection{Classical random walk}

Before we introduce our quantum random walk algorithm, let us first consider classical random walks.

 To perform random walks on a $N$-node graph, we use a simple coin-flipping process with $\mathcal{O}(\log(N))$ time steps. The $\ell$-th bit flips with probability $\sin^2(\theta_\ell/2)$ or does not flip with probability $\cos^2(\theta_\ell/2)$, the total probability being equal to 1. The transition probability matrix elements are given by
\begin{equation}
P^{classical}_{J',J}=\prod_{\ell=0}^{n-1} \Big|\cos^2 (\frac{\theta_\ell}{2})\Big|^{1-i_\ell}\Big|\sin^2 (\frac{\theta_\ell}{2})\Big|^{i_\ell} , \label{clpmatelement}
\end{equation}
where the  $n$-bit binary string state $|I\rangle_c=|i_{n-1},...,i_1,i_{0}\rangle_c$  is determined by $|J'\rangle_c= |I\rangle_c \oplus |J\rangle_c$, where $\oplus$ denotes the bitwise exclusive or (XOR) operation, and the subscript $c$ denotes classical states. The total number of $|\sin^2(\theta_\ell/2)|$, given by $d^{classical}=\sum_{\ell=0}^{n-1}i_\ell $, is the Hamming weight of $|I\rangle_c$, and so corresponds to the Hamming distance between $|J'\rangle_c$ and $|J\rangle_c$ states. This metric measures the number of steps that a walker needs to go from $|J\rangle_c$  to $|J'\rangle_c$  on the Hamming cube.\

 For the four-node graph shown in Fig.~\ref{rw}, the transition probability matrix $\textbf{P}$ for classical random walks reads
\begin{widetext}
\begin{eqnarray}
\textbf{P}^{classical}&=&\begin{bmatrix}
\cos^2(\frac{\theta_0}{2} )\cos^2(\frac{\theta_1}{2}) & \sin^2(\frac{\theta_0}{2} )\cos^2(\frac{\theta_1}{2}) & \cos^2(\frac{\theta_0}{2} )\sin^2(\frac{\theta_1}{2}) &  \sin^2(\frac{\theta_0}{2} )\sin^2(\frac{\theta_1}{2})\\
& \cos^2(\frac{\theta_0}{2} )\cos^2(\frac{\theta_1}{2}) &  \sin^2(\frac{\theta_0}{2} )\sin^2(\frac{\theta_1}{2}) & \cos^2(\frac{\theta_0}{2} )\sin^2(\frac{\theta_1}{2})\\
&& \cos^2(\frac{\theta_0}{2} )\cos^2(\frac{\theta_1}{2}) &\sin^2(\frac{\theta_0}{2} )\cos^2(\frac{\theta_1}{2})\\
&&& \cos^2(\frac{\theta_0}{2} )\cos^2(\frac{\theta_1}{2})\\
\end{bmatrix}\nonumber\\
&=&\begin{bmatrix}
\cos^2(\frac{\theta_1}{2} ) & \sin^2(\frac{\theta_1}{2} )\\
\sin^2(\frac{\theta_1}{2} ) & \cos^2(\frac{\theta_1}{2} )
\end{bmatrix}
\otimes \begin{bmatrix}
\cos^2(\frac{\theta_0}{2} ) & \sin^2(\frac{\theta_0}{2} )\\
  \sin^2(\frac{\theta_0}{2} ) & \cos^2(\frac{\theta_0}{2} )
\end{bmatrix} , \label{clpmatrix4}
\end{eqnarray}
\end{widetext}
where $\otimes$ denotes the Kronecker product.  The lower triangular part of the matrix is omitted due to symmetry.
This simple case demonstrates a general feature for classical transition probability matrix $\textbf{P}^{classical}$: the probability of flipping both bits  is simply a product of the probabilities of flipping the $0$-th bit and the $1$-th bit in arbitrary order. For instance, $P^{classical}_{0,3}=P^{classical}_{|00\rangle,|11\rangle}=\sin^2(\theta_0/2 )\sin^2(\theta_1/2)=P^{classical}_{|00\rangle,|01\rangle}P^{classical}_{|00\rangle,|10\rangle}$; similarly for the other $P^{classical}_{I,J}$'s. The fact that  $\textbf{P}^{classical}$  can be factorized into a Kronecker product of the matrices of each individual  bit indicates that  each bit flips independently, as for a Markovian process.

\subsection{Quantum random walk}

We can  simulate quantum walks \cite{Aharonov1993,Childs2001,Aharonov2001,Moore2001,Szegedy2004,Kendon2006,Childs2017Lecture} on a $N$-node graph to obtain the solution vector $\vec{x}$ from Eq.~(\ref{expand}).  To do it, we use discrete-time coined quantum walk circuit \cite{Kosik2005,Shikano2010}. The circuit for the four-node graph in Fig.~\ref{rw} is  shown in Fig.~\ref{circuit}.  The first two qubits $j_0$ and $j_1$ are state registers that will be initialized to encode the four-node graph, while the third qubit $j_2$ is the coin register.

To derive the quantum transition probability matrix on a graph of $N$ nodes, we consider the state space of the $(n+1)$-qubit circuit as spanned by $\{|i_n\rangle_{\diamond} \otimes |i_{n-1},...,i_1,i_0 \rangle_{q} \}$ with $n=\log_2(N)$: the $(n+1)$-th qubit registers the coin state $|i_n\rangle_{\diamond}$, and the other $n$ qubits encode the $N$-node graph. We take the convention that the rightmost bit is $i_0$. Given a  $n$-bit string  $(j_{n-1},...,j_1,j_0)$, the initialized quantum state reads 
\begin{eqnarray}
|\psi_{0,J} \rangle&=& |0 \rangle_{\diamond} \otimes |j_{n-1},j_{n-2},...,j_2,j_1,j_0 \rangle_{q} \\
&=& |0 \rangle_{\diamond} \otimes|J\rangle_q\,.\nonumber
\end{eqnarray}

Next we let the $|\psi_{0,J} \rangle$ state evolve in random walk: in each walking step, we toss the coin by rotating the coin qubit, and then flip a graph qubit by applying the CNOT gate. This process is repeated on all  the $n$ qubits in the $|j_{n-1},j_{n-2},...,j_2,j_1,j_0 \rangle_{q}$ state, starting with the 0-th qubit.
The corresponding evolution operator reads 
\begin{equation}\label{Ustep}
\mathcal{U}{=}\prod_{k=0}^{n-1}\ ' \Big(|0\rangle_{\diamond} \langle 0|_{\diamond} \otimes 1_{q} + |1\rangle_{\diamond} \langle 1|_{\diamond} \otimes X_k\Big) \cdot \Big(U_3({\bf u}_k)\otimes 1_{q}\Big) ,
\end{equation}
 where the prime ($'$) on the $\prod$ denotes that the $k=0$ operator  applies first to the right, followed by the $k=1$ operator, and so on; the  $1_{q}$ operator is an identity map on the $n$-qubit state $|J\rangle_q$, $X_k$ is a Pauli $X$ gate (the  Pauli matrix $\sigma_x$) that acts on the $k$-th qubit,  and $U_3({\bf u})$ is a single-qubit rotation operator
\begin{equation}
U_3({\bf u})=U_3(\theta,\phi,\lambda)= 
\begin{bmatrix}
\cos(\frac \theta 2) & -e^{i\lambda}\sin(\frac \theta 2)  \\
e^{i\phi} \sin(\frac \theta 2) & e^{i(\lambda+\phi)} \cos(\frac \theta 2)  \\
\end{bmatrix} 
\end{equation}
that acts on the coin qubit state.  Note that the first parentheses in Eq.~(\ref{Ustep}) represents a CNOT gate. It is important to note that here we  use {\it one} quantum coin only to decide on the Pauli $X$ gate operation over all the $n$ qubits, so the order of qubit operations plays a role in the determination of the  transition probability matrix $\textbf{P}$.
   

 The first step is to project $\mathcal{U}$ on $|\psi_{0,J}\rangle$, which leads to
\begin{eqnarray}
\lefteqn{\mathcal{U} |\psi_{0,J} \rangle= \sum^1_{i_{n{-}1},...,i_0=0}  \prod_{\ell=0}^{n-1} U_3({\bf u}_\ell)_{i_{\ell}, i_{\ell-1}}}\label{state}\\
&\times &|i_{n-1}\rangle_{\diamond} \otimes |i_{n-1}\oplus j_{n-1},...,i_1\oplus j_{1},i_0\oplus j_{0}\rangle_{q} \,, \nonumber
\end{eqnarray}
with $i_{-1}= 0$. By tracing out the coin degree of freedom, we obtain the  reduced density matrix for the graph and hence the probability matrix  
$P_{J',J} =  \langle J'| \mbox{Tr}_\diamond [\mathcal{U} |\psi_{0J}  \rangle \langle \psi_{0J} | \mathcal{U}^\dagger ]  |J'\rangle $. The resulting quantum transition probability matrix  elements then read
\begin{eqnarray}
 P^{quantum}_{J',J}&=& \prod_{\ell=0}^{n-1} \Big|U_3({\bf u}_\ell)_{i_{\ell}, i_{\ell-1}}\Big|^2\label{prob}\\
 &=& \prod_{\ell=0}^{n-1} \Big|\cos^2 (\frac{\theta_\ell}{2})\Big|^{1-(i_{\ell}\oplus i_{\ell-1})}\Big|\sin^2 (\frac{\theta_\ell}{2})\Big|^{i_{\ell}\oplus i_{\ell-1}} , \nonumber 
\end{eqnarray}
where $|I\rangle_q=|i_{n-1},...,i_1,i_0 \rangle_{q}$  is determined by $|J'\rangle_q= |I\rangle_q \oplus |J\rangle_q$.  For one $\mathcal{U}$ quantum evolution, the complex phase factors $e^{i\phi_\ell}$ and $e^{i\lambda_\ell}$ play no role. We will see later that these  phases come into play in the case of multiple evolutions $\mathcal{U}^q$.

To understand the transition probability matrix produced by the quantum  walk circuit (Fig.~\ref{circuit}), let us again consider the four-node graph in  Fig.~\ref{rw}, where 
\begin{widetext}
\begin{equation}
\textbf{P}^{quantum}=\begin{bmatrix}
\cos^2(\frac{\theta_0}{2} )\cos^2(\frac{\theta_1}{2}) & \sin^2(\frac{\theta_0}{2} )\sin^2(\frac{\theta_1}{2}) & \cos^2(\frac{\theta_0}{2} )\sin^2(\frac{\theta_1}{2}) & \sin^2(\frac{\theta_0}{2} )\cos^2(\frac{\theta_1}{2})\\
& \cos^2(\frac{\theta_0}{2} )\cos^2(\frac{\theta_1}{2}) & \sin^2(\frac{\theta_0}{2} )\cos^2(\frac{\theta_1}{2})&\cos^2(\frac{\theta_0}{2} )\sin^2(\frac{\theta_1}{2})\\
&& \cos^2(\frac{\theta_0}{2} )\cos^2(\frac{\theta_1}{2}) &\sin^2(\frac{\theta_0}{2} )\sin^2(\frac{\theta_1}{2})\\
&&& \cos^2(\frac{\theta_0}{2} )\cos^2(\frac{\theta_1}{2})\\
\end{bmatrix}. \label{pmatrix4}
\end{equation}
\end{widetext}
Unlike the above classical random walk, this matrix cannot be factorized into a Kronecker product of the matrices of each individual qubit.  The probability  of one qubit flipping depends on the other, indicating that the two qubits are correlated, or in quantum information theory entangled.

\begin{figure}[t]
	\includegraphics[width=0.5\textwidth]{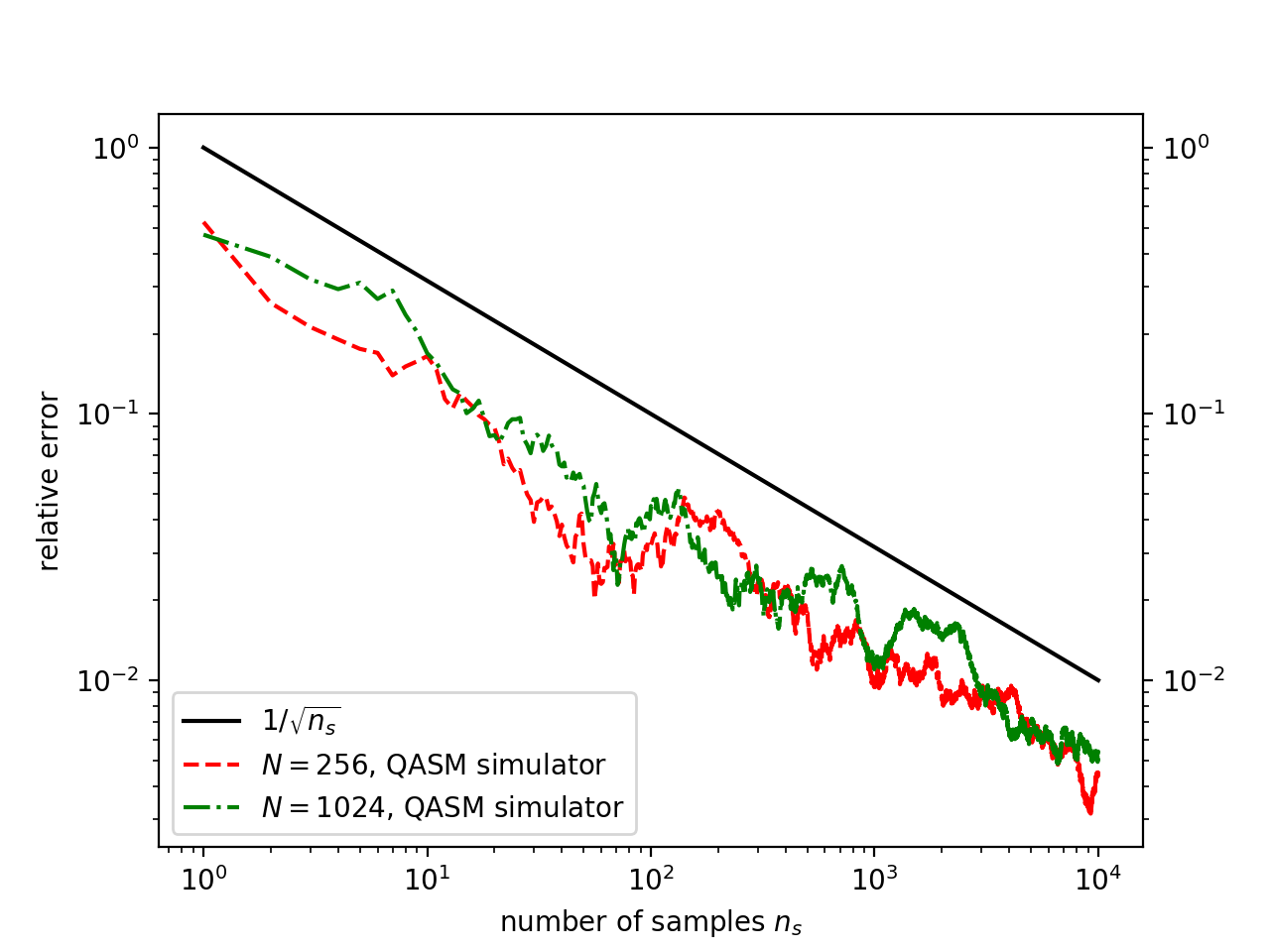} \\
	\includegraphics[width=0.5\textwidth]{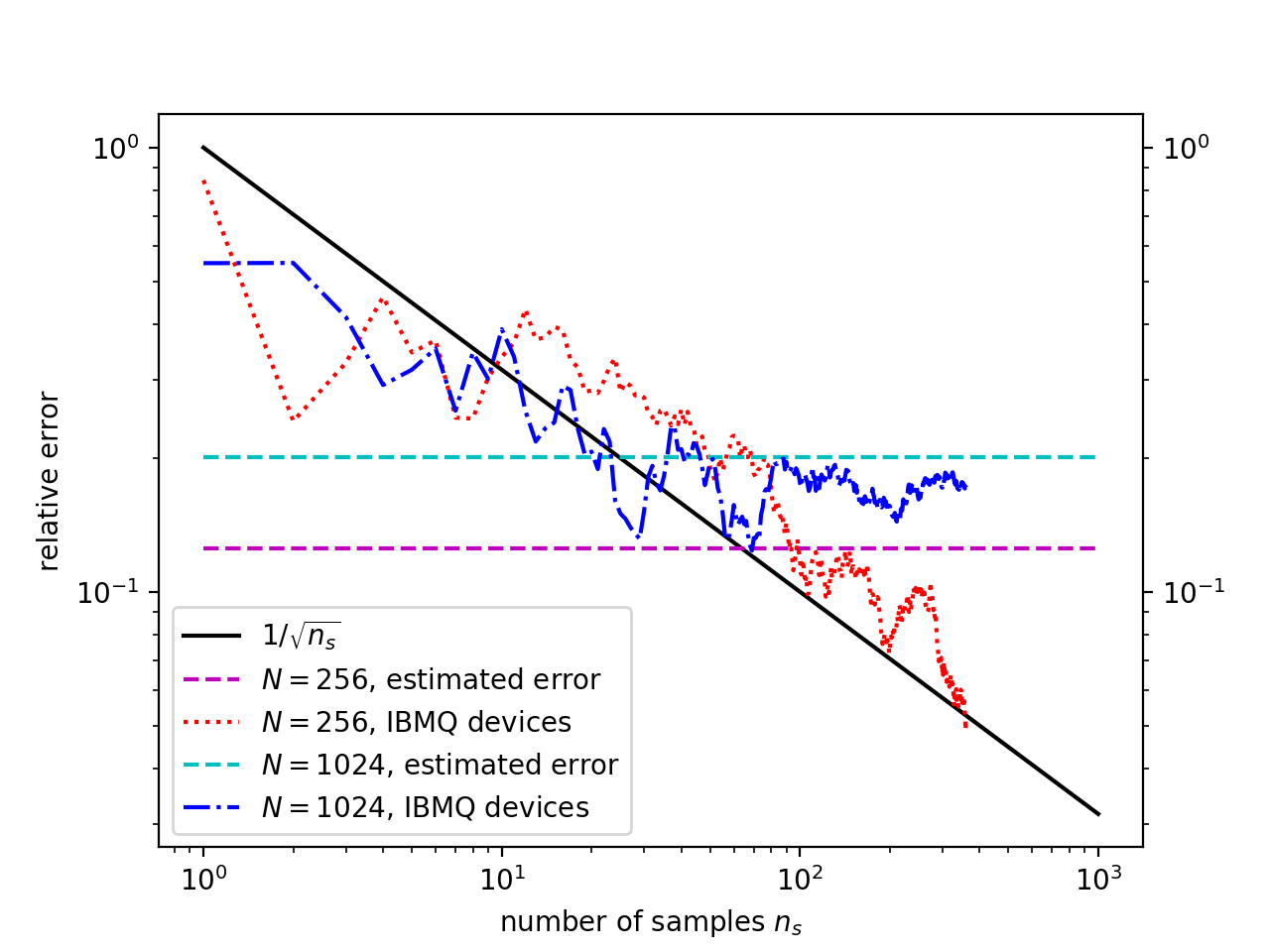}%
	\caption{Relative errors $\epsilon=|x^{exact}_I-x_I|/|x^{exact}_I|$ as a function of the sampling number $n_s$ for  $N=256$ and $N=1024$ matrices. The relevant parameters and  estimated errors for these two matrices can be found in Table \ref{data}.	Black solid lines represent the $1/\sqrt{n_s}$ error reduction expected for Monte Carlo calculations.
(Upper figure) Red dashed line and green dash-dotted line are the results computed by the QASM simulator. (Lower figure) Blue dash-dotted line and red dotted line  are data for the same matrices computed by the IBM Q 20 Tokyo machine or Poughkeepsie machine. Cyan and magenta horizontal dashed lines depict the estimated errors.} \label{errrw}
\end{figure}

In comparison to Eq.~(\ref{clpmatelement}) obtained from the classical random walk, we see that  additional $\mathcal{O}(\log(N))$ XOR operations are required for classical computer to obtain the same quantum transition probability matrix, as can be seen from Eq.~(\ref{prob}). In the case of $N=4$, the classical and quantum transition probability  matrices given by Eqs.~(\ref{clpmatrix4}) and (\ref{pmatrix4}) are related by a permutation 
$\begin{pmatrix}
0&1&2&3 \\
0&3&2&1
\end{pmatrix}$. The quantum  version of the Hamming distance  between $|J\rangle_q$ and $|J'\rangle_q$ is  given by $d^{quantum}=\sum_{\ell=0}^{n-1} i_{\ell}\oplus i_{\ell-1}$, which clearly shows the temporal correlation between the  $\ell$-th and $(\ell-1)$-th qubits. We attribute this correlation to the fact that only {\it one} quantum coin is used to decide on the Pauli $X$ gate over all the $n$ qubits, thus creating some connection between qubits, and to the non-Markovian nature of  quantum walk dynamics \cite{Breuer2016,deVega2017}, in which the quantum circuit memorizes the qubit state $|i_{\ell-1}\rangle$ when it is walking in the direction that has  the  qubit state $|i_\ell\rangle$ in the Hamming cube.


It can be of interest to note that the  circuit given in Eq.~(\ref{Ustep}) is just one possible design leading to a particular correlation between qubits. In general, there are numerous ways to rearrange the walking steps  to obtain different kinds of correlation, and it is possible to design the circuit for specific purposes. A simple  way is to perform the walking steps in Eq.~(\ref{Ustep}) in a reverse order, operating the $k=n-1$ operator to the right first, followed by the $k=n-2$ operator, and so on.  This leads to a different metric $d^{quantum}=\sum_{\ell=0}^{n-1} i_{\ell+1} \oplus i_{\ell}$ with $i_n = 0$. It turns out that this $d^{quantum}$ corresponds to the Hamming distance in the Gray code representation.

The Gray code uses single-distance coding for integer sequence $0 \rightarrow 1 \rightarrow \cdots \rightarrow N-1$, where adjacent integers differ  by single bit flipping.  In the case of the four-node graph  in Fig.~\ref{rw}, the integers $(0,1,2,3)$ in the Gray code representation correspond to the $|00\rangle$, $|01\rangle$, $|11\rangle$, $|10\rangle$ states, respectively. It is obvious that this Gray code representation can be obtained from  the natural binary code representation by a permutation $\begin{pmatrix}
0&1&2&3 \\
0&1&3&2
\end{pmatrix}$. There also exists a permutation that transforms  $\textbf{P}^{classical}$ to $\textbf{P}^{quantum}$  in the Gray code basis. The proof of this correspondence for arbitrary $N$ is given in  Appendix \ref{app:gray}. Both the transform and inverse transform between the natural binary code and Gray code representations take $\mathcal{O}(\log(N))$ operations using classical computer \cite{Knuth2005}. This again shows that the quantum random walk algorithm gains  $\mathcal{O}(\log(N))$ improvement over the classical one.

\begin{table}
	\caption{ Relevant parameters for the matrices $\textbf{A}$ of various sizes used for numerical experiments. Estimated error is defined in the text. \label{data}}
	\begin{ruledtabular}
		\begin{tabular}{c c c c c c}
			N  & c & q& $\gamma$ & Condition number  & Estimated error $\epsilon_0$\\ \hline
			
			64  & 6 & 2 & 0.3 & 1.457 &  \\ \hline
			128  & 6 & 2 & 0.3 & 1.599 &  \\\hline
			256  & 6 & 1 & 0.3 & 1.857 & 0.1255 \\ \hline
			1024  & 10 & 1 & 0.5 & 2.973 & 0.2010 \\ 
		\end{tabular}
	\end{ruledtabular}
\end{table}

 As the change of the Hamming distance for each walking step in the Gray code representation is $\delta d=1$, a quantum  walker in a geodesic of a Hamming cube automatically walks with the least action, that is, with the minimum change of the Hamming distance. This geodesic is a Hamiltonian path on hypercubes \cite{Gilbert1958}.



\begin{figure}[t]
	\includegraphics[width=0.5\textwidth]{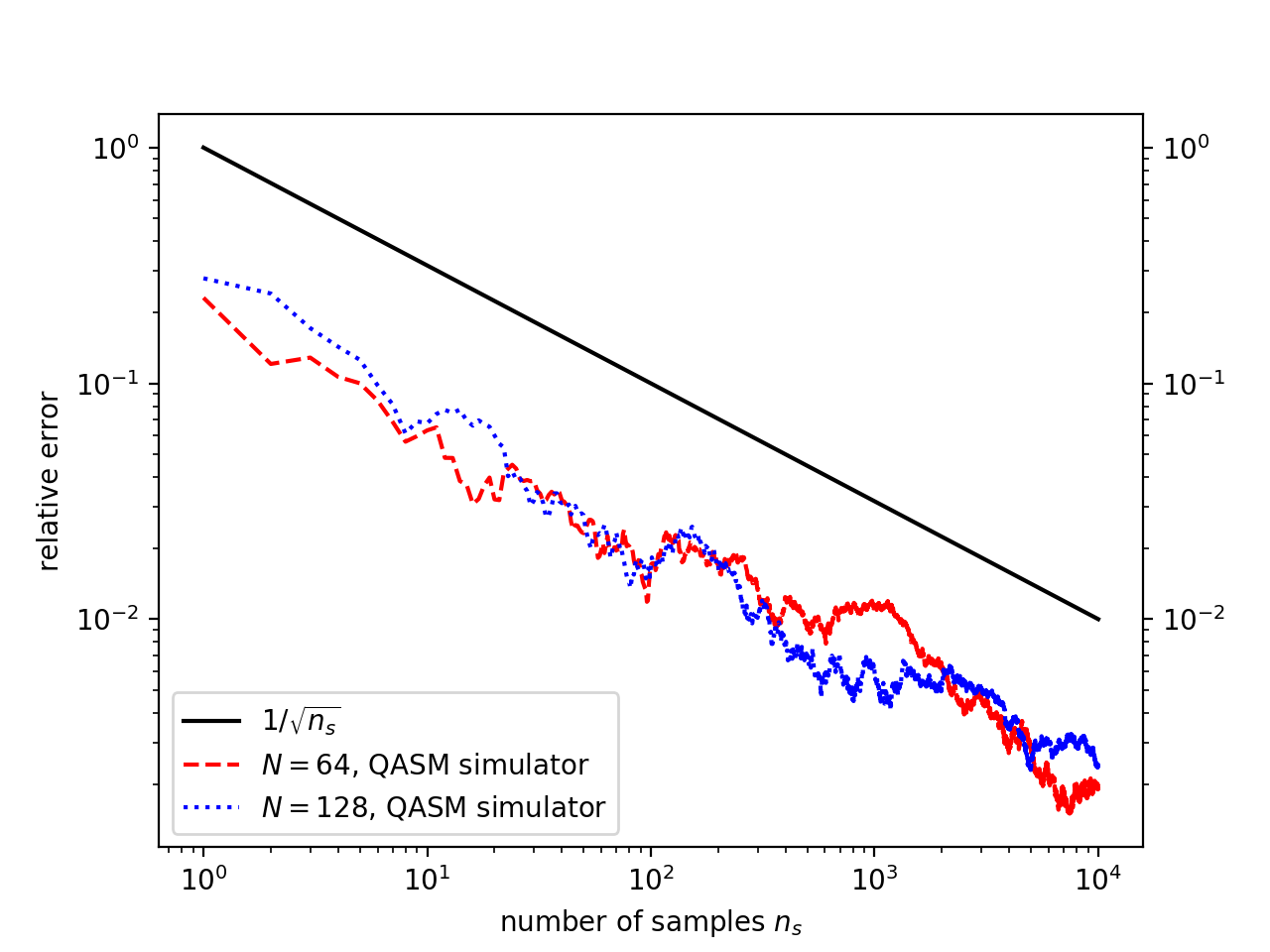} 
	\caption{Relative errors $\epsilon=|x^{exact}_I-x_I|/|x^{exact}_I|$ as a function of the sampling number $n_s$ for  $N=64$ and $N=128$ matrices, obtained by performing two quantum walk evolutions, $\mathcal{U}^2$. Black solid lines represent the $1/\sqrt{n_s}$ error reduction expected for Monte Carlo calculations. Red dashed line and blue dotted line are the results computed by the QASM simulator.} \label{err2}
\end{figure}

It is possible to increase the level of  correlation in the probability matrix by performing multiple quantum evolutions, $\mathcal{U}^q$, where $q$ is the number of quantum walk evolutions. The probability matrix produced by two quantum walk evolutions, $\mathcal{U}^2$, is given by (see Appendix \ref{app:Dertwostep} for derivation)
\begin{eqnarray}
P_{J',J}^{quantum}
&=&\sum_{k=0}^1 \Big|\sum_{I  } f(I, J'\oplus J \oplus I) \delta_{i_{n-1},k} \Big|^2 \,,\label{PJJfortwo}
\end{eqnarray}
where, for $I = (i_{n-1},...,i_0)$ and $K = (k_{n-1},...,k_{0})$, 
\begin{eqnarray}
f(I, K)&=& [U_3({\bf u}_{n-1})]_{i_{n-1},i_{n-2}}\cdots [U_3({\bf u}_{0})]_{i_{0},k_{n-1}} \label{fIK}   \\
&& \times [U_3({\bf u}_{n-1})]_{k_{n-1},k_{n-2}}\cdots[U_3({\bf u}_{0})]_{k_{0},0}  \, ,\notag
\end{eqnarray}
and
\begin{eqnarray}
&& [U_3(\theta,\phi,\lambda)]_{\mu,\nu} \\
&&=  e^{i[\mu \phi + \nu \lambda]} (-1)^{ (1-\mu) \nu} \Big(\cos(\frac{\theta}{2})\Big)^{1-(\mu \oplus \nu)} \Big(\sin(\frac{\theta}{2})\Big)^{\mu \oplus \nu} \label{u3explicit}   \, .\notag
\end{eqnarray}
The fact that the summation over  $I$ in Eq.~(\ref{PJJfortwo}) runs over $\mathcal{O}(2^n)$ state configurations before the square is taken points to the complicated mixing of negative signs and complex phases $\phi_\ell$'s and $\lambda_\ell$'s. The sign problem makes it difficult for pure classical Monte Carlo methods to simulate this transition.

In general, the  dependence of the two-evolution quantum probability matrix on  $\theta_\ell$'s, $\phi_\ell$'s and $\lambda_\ell$'s,  is not trivial. Its explicit expression for the $N=4$ graph is given in Appendix \ref{app:Fournode}. The phases $\phi_\ell$'s and $\lambda_\ell$'s enter into play for graph sizes $N\geq 8$.  On the other hand, the two-evolution probability matrix for classical random walk is given by
\begin{eqnarray}
&&P_{J',J}^{classical} = \\
 && \prod_{\ell=0}^{n-1}  \Big|\cos^{4}(\frac{\theta_\ell}{2})+\sin^{4}(\frac{\theta_\ell}{2})\Big|^{1-i_\ell} \Big|2\cos^{2}(\frac{\theta_\ell}{2}) \sin^{2}(\frac{\theta_\ell}{2}) \Big|^{i_\ell} , \notag 
\end{eqnarray}
which is still factorisable.


\begin{table*}
	\caption{Comparison of various algorithms for solving $N\times N$ linear systems $\textbf{A}\vec{x}=\vec{b}$, with respect to time and space complexities, and Input/Output issues. Note that  for classical Monte Carlo (MC) method, classical random walk (RW) and hybrid quantum random walk (QW), the time complexities in the table are per sampling time. It takes $\mathcal{O}(c n_s)$ samplings to achieve the desired accuracy (see the text).   \label{complexity}}
	\centering
	\begin{ruledtabular}
		\begin{tabular}{c c c c}
			Algorithm & Time & Space for $\textbf{A}$ & Input/Output \\ \hline
			Classical Direct\cite{Golub1996,Saad2003} & $\mathcal{O}(N^3)$ & $\mathcal{O}(N^2)$  & efficient for any $\textbf{A},\vec{x},\vec{b}$ \\ \hline
			Classical Iterative\cite{Golub1996,Saad2003} & $\mathcal{O}(N^2)$ & $\mathcal{O}(N^2)$  & efficient for any $\textbf{A},\vec{x},\vec{b}$ \\ \hline
			Quantum HHL\cite{Harrow2009} & $\mathcal{O}(\log (N))$ & $\mathcal{O}(\log (N))$ qubits & norm $||\vec{x}||$  not available \\ 
			& &  & difficult for $\textbf{A},\vec{x},\vec{b}$ \\ \hline
			Classical MC\cite{Forsythe1950,Barto1993,Lu2003} &  $\mathcal{O}(N)$  & $\mathcal{O}(N)$ & efficient for any $\vec{x},\vec{b}$ \\
			(for one component $x_I$) &  &  &  limited $\textbf{A}$ (stochastic $\textbf{P}$)\\ \hline
			Classical RW on HC &$\mathcal{O}(\log (N))$  & $\mathcal{O}(\log (N))$  & efficient for any $\vec{x},\vec{b}$\\
			(for one component $x_I$) & &  & limited $\textbf{A}$ (factorisable $\textbf{P}$) \\ \hline
			Hybrid QW on HC & $\mathcal{O}(\log (N))$    & $\mathcal{O}(\log (N))$ qubits & efficient for any $\vec{x},\vec{b}$\\
			(for one component $x_I$) & &  & limited $\textbf{A}$ (correlated $\textbf{P}$)\\
		\end{tabular}
	\end{ruledtabular}
\end{table*}

\section{Numerical results}

Figure \ref{errrw} shows  the performance of our hybrid quantum random walk algorithm on  linear systems  of dimension $N=256$ and $N=1024$. Their relative errors decrease with increasing sampling number. The relative error is defined as $\epsilon=|x^{exact}_I-x_I|/|x^{exact}_I|$ for the $I$-th component of the solution vector $\vec{x}$, where $\vec{x}^{exact}$ is the exact result obtained with the NumPy package. To demonstrate, we use randomly generated   vectors $\vec{b}$ and matrices $\textbf{A}$ with a uniform distribution, $b_I \in [-1,1]$ and  $\theta_\ell \in [0,\pi]$.  We choose $\gamma$ and $c$ such that the error introduced by the Neumann expansion is within $\mathcal{O}(10^{-4})$. See Table \ref{data} for the relevant parameters of the two matrices. The program is written and compiled with Qiskit version 0.7.2. The simulation results (upper figure) are obtained using QASM simulator \cite{IBMQ2017qiskit}, while the quantum machine results (lower figure) are obtained using IBM Q 20 Tokyo device or Poughkeepsie device \cite{IBMQdata2019,IBMQnews2019}.

The  curves obtained by the QASM simulator are results averaged over ten runs. Their relative errors decrease as $1/\sqrt{n_s}$, where  $n_s$ is the number of random walk samplings. This $1/\sqrt{n_s}$ reduction  is typical of Monte Carlo simulations, because the hybrid quantum walk algorithm has essentially the same structure as  classical Monte Carlo methods. So, we do not gain any speedup in  sampling number. Yet, this result substantiates the fact that our proposed algorithm works on ideal quantum computers.

For real IBM Q quantum devices,  the accuracy stops improving after a certain number of samplings (see the plateau (blue dash-dotted curve) and oscillation (red dotted curve) in Fig.~\ref{errrw}). This hardware limitation can be estimated using an error  formula $\epsilon_0 \sim \kappa \times E_r$, where $\kappa$ is the condition number for the matrix $\textbf{A}$ and $E_r$ is the readout error of real machines.   The condition number $\kappa$ gauges the ratio of the relative error in the solution vector $\vec{x}$ to the relative error in the $\textbf{A}$ matrix \cite{Saad2003}: some perturbation in the matrix, $\textbf{A}+\delta \textbf{A}$, can cause an error in the solution vector, $\vec{x}+\delta \vec{x}$, such that $||\delta \vec{x}||\sim \kappa \times ||\delta \textbf{A}||$.  By taking $E_r$ as an estimate for $||\delta \textbf{A}||$, we obtain the above error  for the solution vector as $\epsilon_0 = ||\delta \vec{x}||\sim \kappa \times E_r$.
The condition numbers given  in Table \ref{data} are computed by using Eq.~(\ref{prob}) to construct the $\textbf{A}$ matrices. For the  average readout error of IBM Q 20 Tokyo device, we use $E_r=6.76\times 10^{-2}$ \cite{IBMQdata2019}. The estimated errors $\epsilon_0$ are given in Table \ref{data}. We see that the relative errors fall below the  respective errors, indicating that the precision limit is due  to the readout error of the current NISQ hardware. Note that the machines are calibrated several times during data collection, so the hardware error varies and the $E_r$ value is only an estimate.

 Figure \ref{err2} shows the results for linear systems of dimension $N=64$ and $N=128$, obtained by the  QASM simulator that performs two quantum walk evolutions with uniformly distributed $(\theta_\ell,\phi_\ell,\lambda_\ell)\in[0,\pi]$. The relevant parameters for these two matrices are given in Table \ref{data}. The results again evidence that the algorithm works well, even in the presence of complex phases $\phi_\ell$'s and $\lambda_\ell$'s.

The communication latency between classical and quantum computer is the most time-consuming part, containing $\mathcal{O}(cn_s)$ communications. Fortunately, this number does not scale as $N$. For users with direct access to the quantum processors,  communication bottleneck should be less severe.


\section{Discussions}

A comparison of  computational  resources is given in Table \ref{complexity}. For  hybrid quantum walk algorithm, we need $1+\log(N)$ qubits, $q\log(N)$ CNOT gates, and $q\log(N)$ $U_3$ gates, where $q$ is the number of evolutions. The initialization takes $\log(N)$ $X$ gates; but since they can be executed simultaneously, the initialization occupies one time slot only. Totally $1+2q\log(N)$ time slots are required for one quantum walk evolution to obtain one component of the solution vector $\vec{x}$. This can be an advantage when one is interested in  partial information about $\vec{x}$. \

The same  amount of time slots can be similarly derived for the classical random walk algorithm. Yet, we stress that these two algorithms deal with different transition probability matrices: factorisable matrices for classical random walk, and more complex correlated matrices for quantum random walk. The qubit-qubit correlation built into the correlated matrix can potentially be harnessed to perform complex tasks.

Other advantages of the algorithms we propose are: 

\noindent{\bf(i)} By restricting the matrices $\textbf{A}$ to those that can be encoded  in Hamming cubes, we can sample both classical and quantum  random walk spaces that scale exponentially with the number of bits/qubits, and hence gain space complexity.\

\noindent{\bf(ii)} Classical Monte Carlo methods have time complexity of $\mathcal{O}(N)$ for general $\textbf{P}$ matrices. For the  matrices here considered, our algorithms have $\mathcal{O}(\log (N))$.

\noindent{\bf(iii)} It is easier to access input and output than the  HHL-type algorithm. 

\noindent{\bf(iv)} Random processes in a quantum computer are  fundamental, and so are not plagued by various problems associated with pseudo-random number generators \cite{Srinivasan2003}, like periods and unwanted correlations.

\noindent{\bf(v)} Our quantum  algorithm  can run on noisy quantum computers whose coherence time is short.

\section{Conclusion}

We propose a hybrid quantum algorithm suitable for NISQ quantum computers to solve systems of linear equations. The solution vector $\vec{x}$ and constant vector $\vec{b}$ we consider here are classical data, so the input and readout  can be executed easily.
Numerical simulations using IBM Q systems support the feasibility of this algorithm.  We demonstrate  that, by performing two quantum walk evolutions, the resulting probability matrix become more correlated in the parameter space.  As long as the quantum circuit in this framework produces  highly correlated probability matrix  with a relatively short circuit depth, we can always gain quantum advantages over classical circuits. 

\begin{acknowledgments}
We thank Chia-Cheng Chang, Yu-Cheng Su, and Rudy Raymond for discussions. Access to IBM Q systems is provided by IBM Q Hub at  National Taiwan University. This work is supported in part by Ministry of Science and Technology, Taiwan, under grant No. MOST 107-2627-E-002 -001 -MY3, MOST 106-2221-E-002 -164 -MY3, and MOST 105-2221-E-002 -098 -MY3.
\end{acknowledgments}

%
%
%

\appendix


\renewcommand{\theequation}{\mbox{A.\arabic{equation}}} 
\setcounter{equation}{0} %

\newtheorem{lemma}{Lemma}
\newtheorem{theorem}{Theorem}

\section{Gray code basis\label{app:gray}}
The natural binary code $B=(B_{n-1},B_{n-2},...,B_1,B_0)$ is transformed to the Gray code basis \cite{Knuth2005} according to
\begin{equation}
g(B)_i=B_{i+1}\oplus B_i ,
\end{equation}
$\forall i \in \{0,...,n-1 \}$ with $B_n=0$. The probability matrix in the Gray code basis is given by
	\begin{eqnarray}
	P^{quantum}_{J',J}&=& \prod_{\ell=0}^{n-1} \Big|U_3(\theta_\ell)_{i_{\ell}, i_{\ell+1}}\Big|^2\label{prob2}\\
	&=& \prod_{\ell=0}^{n-1} \Big|\cos^2 (\frac{\theta_\ell}{2})\Big|^{1-(i_{\ell}\oplus i_{\ell+1})}\Big|\sin^2 (\frac{\theta_\ell}{2})\Big|^{i_{\ell}\oplus i_{\ell+1}}  \nonumber 
	\end{eqnarray}
 with $i_N=0$.

\begin{lemma} \label{lemma1}
	Let $S_N$ be the set of all possible n-bit strings $\{ (S_{n-1},S_{n-2},...,S_1, S_0) | S_i \in \{0,1\}~ \forall~ i  \in \{0,1,...,n-1 \} \}$ with $n=\log_2 N$, and $\pi$ be a permutation of the set $S_N$. If there exists a function $f:S_N \mapsto \mathbb{R}$ such that for $A \in \mathbb{R}^{N\times N}$,
	\begin{equation} \label{condition1}
	A_{I\oplus J,J}=f(I)
	\end{equation}
	$\forall~ I,J \in S_N$, and if $\pi$ is bitwise XOR homomorphic, then we have
	$A_{\pi(I\oplus J),\pi(J)}=f(\pi(I))$.
\end{lemma}

\begin{proof}
Since $\pi$ is bitwise XOR homomorphic, 
 Eq.~(\ref{condition1})  leads to
\begin{eqnarray}
A_{\pi(I\oplus J),\pi(J)}&=& A_{\pi(I)\oplus \pi(J),\pi(J)}\notag\\
&=&f(\pi(I)) 
\end{eqnarray}
$\forall~ I,J \in S_N$. 
\end{proof}

\begin{lemma} \label{lemma2}
	Let $B \in S_N$ be represented by $(B_{n-1},...,B_0)$.
	Let $g:S_N \mapsto S_N$ be a function that transforms from natural bit string  to Gray code according to $g(B)_i=B_{i+1}\oplus B_i$, $\forall~ i \in \{0,1,...,n-1 \} $ with $B_{n}= 0$. Then $g$ is a bitwise XOR homomorphism.
\end{lemma}

\begin{proof}
	Let $I,J \in S_N$ be represented by bit strings $(I_{n-1},...,I_0)$ and $(J_{n-1},...,J_0)$, respectively. Using
	\begin{eqnarray}
	g(I)_i&=&I_{i+1}\oplus I_i\notag\\
	g(J)_i&=&J_{i+1}\oplus J_i
	\end{eqnarray}
	with $I_n=J_n=0$, we get
	\begin{eqnarray}
	[g(I)\oplus g(J)]_i  &=&g(I)_i \oplus g(J)_i\notag\\
	&=&(I_{i+1} \oplus I_{i}) \oplus (J_{i+1} \oplus J_{i})\notag\\
	&=&(I_{i+1} \oplus J_{i+1}) \oplus (I_{i}  \oplus J_{i})\notag\\
	&=&g(I\oplus J)_i .
	\end{eqnarray} 
	$\forall~ i \in {0,...,n-1}$.
\end{proof}

Using \textbf{Lemma \ref{lemma1}} and \textbf{Lemma \ref{lemma2}}, the following theorem is clear.
\begin{theorem}
There exists a permutation that maps the probability matrix produced by classical random walk to the probability matrix  given in Eq.~(\ref{prob2}) produced by the quantum random walk circuit in a reverse order, that is, in Gray code basis.
\end{theorem}

\renewcommand{\theequation}{\mbox{B.\arabic{equation}}} 
\setcounter{equation}{0} %

\section{Derivation of Eq.~(\ref{PJJfortwo})\label{app:Dertwostep}}

We use the evolution operator given in Eq.~(\ref{Ustep}),
\begin{equation}
 \mathcal{U} 
  = \!\!\!\! \sum_{i_{n-1},...,i_0,i_{-1}} \prod_{\ell=0}^{n-1}  [U_3({\bf u}_\ell)]_{i_\ell,i_{\ell-1}} (X_\ell)^{i_\ell} ( | i_{n-1} \rangle_\diamond   \langle i_{-1} |_\diamond )  
\end{equation}
to compute the two-evolution operator
\begin{eqnarray}
   \mathcal{U}^2 \!\!\!
&=&\!\!\! \sum_{\substack{i_{n-1},...,i_0,i_{-1}\\k_{n-1},...,k_0,k_{-1}}} \prod_{\ell=0}^{n-1}  [U_3({\bf u}_\ell )]_{i_\ell,i_{\ell-1}}  [U_3({\bf u}_\ell )]_{k_\ell,k_{\ell-1}} \notag \\
&& \times (X_\ell)^{i_\ell + k_\ell } \delta_{i_{-1},k_{n-1}} ( | i_{n-1} \rangle_\diamond   \langle k_{-1} |_\diamond )  \,. 
\end{eqnarray}

Next we project the $\mathcal{U}^2$ operator  on the  $|\psi_{0J}  \rangle$ state,
\begin{eqnarray}
\mathcal{U}^2|\psi_{0J}  \rangle  &=&   \sum_{\substack{i_{n-1},...,i_0,i_{-1}\\k_{n-1},...,k_0,k_{-1}}} \prod_{\ell=0}^{n-1}  [U_3({\bf u}_\ell )]_{i_\ell,i_{\ell-1}}  [U_3({\bf u}_\ell )]_{k_\ell,k_{\ell-1}}  \notag \\
 &&\times  \delta_{i_{-1},k_{n-1}} \delta_{0,k_{-1}}  | i_{n-1} \rangle_\diamond | I \oplus K \oplus J \rangle_q    \notag \\
  &=&  \sum_{I,K}   f(I,K)   | i_{n-1} \rangle_\diamond     | I \oplus K \oplus J \rangle_q\, , 
\end{eqnarray}
where $f(I,K)$ is given in Eq.~(\ref{fIK}) and 
\begin{equation}
| I \oplus K \oplus J \rangle_q= | i_{n-1} \oplus k_{n-1} \oplus j_{n-1},..., i_{0} \oplus k_{0} \oplus j_0\rangle_q\,.\notag
\end{equation}
We then project $\mathcal{U}^2|\psi_{0J}  \rangle $ on the final state $| k \rangle_\diamond | J' \rangle_q$
\begin{eqnarray}
\langle J'|_q \langle k|_\diamond \mathcal{U}^2|\psi_{0J}  \rangle& =&  \sum_{I,K}   f(I,K)   \delta_{k, i_{n-1}}      \delta_{J', I \oplus K \oplus J } \notag \\
& = & \sum_{I}   f(I,I\oplus J' \oplus J)   \delta_{k, i_{n-1}}  ,   
\end{eqnarray}
which leads to the  probability matrix  elements as
\begin{eqnarray}
P_{J',J} &&=  \langle J'| \mbox{Tr}_\diamond [\mathcal{U}^2 |\psi_{0J}  \rangle \langle \psi_{0J} | (\mathcal{U}^\dagger)^2 ]  |J'\rangle  \notag \\
&&= \sum_k \Big|\langle J'|_q \langle k|_\diamond \mathcal{U}^2 |\psi_{0J}  \rangle \Big|^2 \notag \\
&&= \sum_{k=0}^1 \Big|\sum_{I  } f(I, J'\oplus J \oplus I) \delta_{i_{n-1},k} \Big|^2 .
\end{eqnarray}

\renewcommand{\theequation}{\mbox{C.\arabic{equation}}} 
\setcounter{equation}{0} %

\section{Two-evolution quantum walk on $N=4$ graph\label{app:Fournode}}

The probability matrix elements $P^{quantum}_{J',J}$ for two quantum evolutions $\mathcal{U}^2$ on the four-node graph read 
\begin{eqnarray}
P_{00}&=&P_{11}=P_{22}=P_{33} \\
&=&\frac{1}{4}\sin^2\theta_0+\frac{1}{8}(1+\cos^2\theta_1 +\cos^2\theta_0 \notag  \\
&& +4 \cos \theta_1 \cos \theta_0 +\cos^2 \theta_1 \cos^2 \theta_0 -\sin^2 \theta_1 \sin^2 \theta_0) ,\notag\\
P_{01}&=&P_{23}=\frac{1}{4}\sin^2\theta_1,\\
P_{02}&=&P_{13}=\frac{1}{4}\sin^2\theta_1,\\
P_{03}&=&P_{12}=\frac{1}{4} \Big(1-2\cos\theta_1\cos\theta_0+\cos^2\theta_1 \Big) \, .
\end{eqnarray} 
Surprisingly, in this case the matrix elements do not depend on the $(\phi_1,\phi_2)$ and $(\lambda_1,\lambda_2)$ phases. However,  the matrix elements do depend on complex phases when $N\geq 8$, as can be numerically  checked. Note that $(P_{01},P_{02},P_{23},P_{13})$ depend on $\theta_1$ only: the destructive interference between configurations totally eliminates the $\theta_0$ dependence, which  is difficult to do by simple classical random walks.

\bibliography{rwmi}

\end{document}